\documentclass[12pt]{article}
\usepackage{setspace}
\usepackage{latexsym,amsthm,amsmath,amssymb,url}
\usepackage{enumitem}
\usepackage[round]{natbib}
\usepackage{eurosym}

\usepackage{bbding}
\usepackage{graphics}
\usepackage[usenames]{color}

\usepackage{tikzsymbols}

\theoremstyle{definition}
\newtheorem{lemma}{Lemma}

\theoremstyle{definition}
\newtheorem{definition}{Definition}
\theoremstyle{definition}

\theoremstyle{definition}
\newtheorem*{theorem*}{Theorem}
\theoremstyle{definition}
\newtheorem{theorem}{Theorem}
\theoremstyle{definition}

\newcommand{\abs}[1]{|#1|}
\newcommand{\set}[1]{\left\{#1\right\}}

\newcommand{\ceiling}[1]{\lceil#1\rceil}

\newcommand{\Workers}{W}
\newcommand{\worker}[1]{w_{#1}}
\newcommand{\Firms}{F}
\newcommand{\firm}[1]{f_{#1}}

\newcommand{\match}{\mu}

\newcommand{\reduce}{\mathcal{R}}

\newcommand{\WithColor}[1]{#1} \newcommand{\WithoutColor}[1]{}       
\newcommand{\matchingC}[1]{yellow}

\usepackage[pdftex, citecolor=blue, colorlinks]{hyperref}

\begin{document}

\title{Note on the size of a stable matching\thanks{We thank Vince Crawford, Gabriel Facchini, Evan Piermont, Michael Richter, Frances Ruane, Joel Sobel, Ija Trapeznikova, and Gregory F.\ Veramendi for useful discussions. An associate editor and two anonymous referees made excellent suggestions that improved the paper considerably. We are indebted to them for suggesting the analysis of Section~\ref{sec:relationship}.}}
\author{
Gregory Z. Gutin\footnote{Computer Science Department, Royal Holloway University of London.} \hspace{.1in} Philip R.\ Neary\footnote{Economics Department, Royal Holloway University of London.} \hspace{.1in} Anders Yeo\footnote{IMADA, University of Southern Denmark.} $^{,}$\footnote{Department of Mathematics, University of Johannesburg.}
}

\date{\today}

\maketitle

\begin{abstract}
\noindent
Consider a one-to-one two-sided matching market with workers on one side and single-position firms on the other, and suppose that the largest individually rational matching contains $n$ pairs.
We show that the number of workers employed and positions filled in every {stable matching} is bounded from below by $\ceiling{\frac{n}{2}}$ and we characterise the class of preferences that attain the bound.
We then identify the minimum number of {\it equilibrium pairings} that must be ``sacrificed'' when maximising the employment rate is the objective; if each stable matching is of size $\ceiling{\frac{n}{2}}$, then no such pairs appear when all vacancies are filled.
\end{abstract}


\newpage

\section{Introduction}\label{sec:intro}

Our starting point in this paper is a one-to-one two-sided matching market, {\`a} la \cite{GaleShapley:1962:AMM}, with ``workers'' on one side and single-position ``firms'' on the other.\footnote{While we employ labour market terminology (see the classic papers of \cite{CrawfordKnoer:1981:E} and \cite{KelsoCrawford:1982:E}), the set up is applied widely and we could have labelled the two sides ``students'' and ``colleges'', or ``men'' and ``women'', etc. \cite{Roth:2015:} is a popular book that details many environments to which the framework applies.}
The only assumption made on preferences is that the largest individually rational matching is of size $n$.
That is, there exists an assignment in which each of the $n$ matched workers prefers their position to being unemployed and each of the $n$ firms that has its post filled prefers its employee over leaving the position empty.
Given this, we ask:\ what fraction of workers are assured of employment in equilibrium? 

Our first result is that the number of those employed in equilibrium can never go below $50\%$ of what is attainable.
In other words, every stable matching is at least half the size of a maximum matching:

\begin{theorem}\label{thm:main}
Fix a one-to-one two-sided matching market with an individually rational matching of size $n$.
Then, every stable matching has at least $\ceiling{\frac{n}{2}}$ pairs.
\end{theorem}

Formulating the stable matching problem using undirected graphs facilitates a short proof.\footnote{This formulation appears in, for example, Chapter 1 of the \citeyear{EcheniqueImmorlica:2023:} edited volume ``Online and Matching-Based Market Design'' (\citeauthor{EcheniqueImmorlica:2023:}).}
Fix a bipartite graph, $G = (W, F; E)$, where $W$ and $F$ are the set of workers and the set of firms respectively, and the edge set, $E$, represents the acceptable pairs (i.e., for each edge, both parties prefer to be matched to each other over being unmatched).
Finally, assume that each vertex has a strict preference ranking over each of its neighbours.

\begin{proof}
Confirming Theorem~\ref{thm:main} requires showing that (i) a stable matching is a maximal matching, and (ii) a maximal matching is at least half the size of a maximum matching.\footnote{Both statements are likely folklore, but since we are unable to find a reference for either we include a proof of each.}

\begin{enumerate}[label=(\roman*)]
\item
A stable matching is maximal:\\
Let $\mu^*$ denote a stable matching, and suppose, towards a contradiction, that $\mu^*$ is not maximal.
Given this, there exists an edge in the graph $G$ with neither end-point in $\mu^*$.
But because any such edge corresponds to an acceptable pair, call it $(w, f)$, where both parties are unmatched in $\mu^*$.
The pair $(w, f)$ is therefore blocking for $\mu^*$ and thus contradicts that $\mu^*$ is stable.

\item
A maximal matching is at least half the size of a maximum matching:\\
Let $\mu_{\max}$ denote a maximum matching of size $n$, let $\mu$ denote a (different) maximal matching, and suppose that $\abs{\mu} < \ceiling{\frac{n}{2}}$.
For every edge $(w, f)$ in $\mu$, there exists at least one and at most two edges in $\mu_{\max}$ that have $w$ or $f$ as an end-point.
Therefore, $2\abs{\mu} < \abs{\mu_{\max}}$ means that there is some edge in $\mu_{\max}$ such that neither end-point is an end-point to an edge in $\mu$, which contradicts that $\mu$ is maximal.
\end{enumerate}
\end{proof}

The above proof is straightforward but not overly informative.
While a simple example shows that the bound is tight,\footnote{We are grateful to an anonymous referee for providing the following example. To show the result for any even $n$, begin with $n = 2$. Preference lists are ordered such that the most favoured partners appear first and so on, and missing entries indicate unacceptable partners. They are as follows. $\worker{1}: (\firm{1}, \firm{1}); \, \worker{2}:(\firm{1}); \, f_1 : (\worker{1}, \worker{2}); \, \firm{2}: (w_1)$. The only stable matching is of size 1 and pairs worker $\worker{1}$ with $\firm{1}$, but there is an individually rational matching of size two given by $\set{(\worker{1}, \firm{2}), (\worker{2}, \firm{1})}$. This works for $n=2$, but note that disjoint copies of the example work for any even $n$. For $n$ odd, simply use the ``even construction'' from above up to $n-1$, and then add one more worker $w^*$ and one more firm $f^*$ such that the pair $(w^*, f^*)$ is the only acceptable pair for each.} a natural question remains unanswered:~what is the structure of preferences that attain it?

Theorem~\ref{thm:characterize3} provides the answer.
All preferences that attain the bound share a property:~amongst those participants that do not appear in an individually rational matching, and therefore also not in any stable one, at least one must find the other unacceptable.
If not, then there would be an acceptable pair amongst them which would increase the size of a feasible assignment.
One natural (sub)class of preferences that attain the bound also have a feature that could be termed ``agreement on those at the top''.
That is, every worker that appears in a stable matching ranks the firms that do not appear in a stable matching below the firms that do (with a mirror-image statement holding for each firm that appears in a stable matching).

It is worth emphasising that Theorem~\ref{thm:main} does not guarantee that the equilibrium {\it employment rate}, the number of workers with jobs as a percentage of the total labour force, must be at least 50\%.
Far from it.
Indeed, Theorem~\ref{thm:characterize3}'s characterisation illustrates how any two-sided market could, in theory, be extended indefinitely so that the equilibrium employment rate can be taken arbitrarily close to zero even when there unemployed workers and unfilled jobs remain.
This kind of market failure will maintain as long as either unemployed workers prefer unemployment over any of the empty vacancies or the opposite (or both).

Our final result explores the overlap in the set of pairs that appear in some stable matching, that we term {\it equilibrium pairings}, and the collection of pairs that appear in a maximum (individually rational) matching.
The precise statement, Theorem~\ref{thm:overlap}, shows that the relationship is monotonic (in fact linear):\ when a stable matchings is as small as can be relative to a maximum matching $\big($$=$$\ceiling{\frac{n}{2}}$, where $n$ is the size of a maximum matching$\big)$, no stable pairings appear in a maximum matching; conversely, when a stable matching is maximum, all equilibrium pairings may appear.
In particular, it allows one to identify a lower bound on the number of {equilibrium pairings} that must be ``sacrificed'' if maximising the employment rate is the objective.
This confirms an additional cost of filling all vacancies that occurs at a more granular level than simply the reduction in the number of worker-firm pairs.

In a big picture sense, our analysis is about the cost of stability, i.e., what demands stability places relative to what is possible without insisting on it. 
That being said, the measures we employ are the size of a matching --- the number of worker-firm pairs --- and the number of equilibrium pairings.
One thing our analysis is silent on is precisely how well off are the market participants that are paired, as in more standard measures like {\it the price of anarchy} and {\it the price of stability} that were applied to matching markets in \cite{BoudreauKnoblauch:2013:TD,BoudreauKnoblauch:2014:MSS}.\footnote{The price of anarchy and the price of stability were first introduced for non-cooperative games in \cite{KoutsoupiasPapadimitriou:1999:} and \cite{SchulzStier:2003:SODA} respectively.}
Increasing the size of a matching or breaking up equilibrium pairings may means that more participants fare better than their unmatched option, but if it comes at the cost of ``happier pairings'' then the effect on overall welfare is unclear.
What kind of trade offs exist between these two objectives seems worthy of exploration, but we leave this to future work.

Before beginning the paper proper, it is worth reiterating that our analysis holds only for the simplest of matching environments:~one-to-one two-sided markets.
Whether our results carry over to other, richer matching environments seems an interesting avenue to pursue but is beyond the scope of this short paper.\footnote{Other environments include one-sided markets, the so-called room-mate problem, and many-to-one two-sided markets with ``traditional'' preferences and ``richer preferences'' (for example, there may exist couples with shared preferences \citep{Roth:1984:JPE,KlausKlijn:2005:JET,KlausKlijn:2009:JET}, schools may have preferences over collections of students \citep{BiroHassidim:2025:RESTAT}, and so on). A variant on our question of interest can be asked by comparing the surplus in equilibrium with the total surplus attainable in the classic assignment game of \citeauthor{ShapleyShubik:1971:IJGT}.}

Section~\ref{sec:notation} defines the environment and fixes notation.
Despite having stated Theorem~\ref{thm:main} in the language of undirected graphs, in Section~\ref{sec:notation} we change tack and adopt the directed graph (digraph) formulation of the stable matching problem first proposed in \cite{Maffray:1992:JCTB}.
Our reason for settling on this approach is that it provides a nice visual aid to illustrate preferences that attain the bound.
In Section~\ref{sec:characterisation}, we characterise this class using digraphs.
In Section~\ref{sec:relationship}, we revert to the undirected graph formulation and illustrate the relationship in the maximum ``overlap'' in pairs that appear in a stable matching and a maximum matching.


\section{Preliminaries}\label{sec:notation}

The set up is the original one-to-one two-sided matching model of \citeauthor{GaleShapley:1962:AMM} enriched in two ways:\ (i) there need not be the same number of participants on each side of the market, and (ii) there may be those who prefer to be unmatched over being paired with some participants.

Let $W$ denote the finite set of {\it workers} and let $F$ denote the finite set of single-position {\it firms}.
The sets $W$ and $F$ are disjoint.
Each market participant has a complete, transitive, and strict preference ordering over some nonempty subset of those on the other side of the market.

A {\it matching} is a one-to-one mapping $\mu$ from $W \cup F$ to itself, such that:\ $\mu(w) = f$ if and only if $\mu(f) = w$, in which case we will say that $w$ is {\it matched} to $f$; if $\mu(w)$ is not in $F$, then $\mu(w) = w$, in which case $w$ is {\it unmatched}; if $\mu(f)$ is not in $W$, then $\mu(f) = f$, in which case $f$ is said to be {\it unmatched}.

A matching is {\it individually rational} if no participant prefers to be unmatched over who they are matched to.
We say that the pair $(w,f)$ in $W \times F$ is {\it acceptable} if $w$ and $f$ prefer each other over being unmatched, and we assume that every market participant belongs to at least one acceptable pair.
A matching $\mu$ is {\it stable} if it is individually rational and there is no {\it blocking pair:}\ an acceptable pair $(w, f)$ such that worker $w$ is either unmatched or prefers firm $f$ to $\mu(w)$ and firm $f$ is either unmatched or prefers $w$ to $\mu(f)$.

Given a matching market, we define the associated {\it matching digraph}, $D = \big(V, A)$, where the vertex set $V$ is the collection of acceptable pairs and the arc set $A$ is defined as follows.\footnote{The digraph terminology we employ is the same as that in \cite{GutinNeary:2023:GEB,GutinNeary:2024:GEB}.}
\begin{center}
$\begin{array}{rcl}
  A_{\Workers} & := & \{(\worker{}, \firm{i}) (\worker{}, \firm{j}) \; | \; \worker{} \text{ prefers } \firm{j} \text{ to } \firm{i} \} \\
  A_{\Firms}  & := &  \{ (\worker{i}, \firm{}) (\worker{j}, \firm{}) \; | \; \firm{} \text{ prefers } \worker{j} \text{ to } \worker{i} \} \\
  A & := & A_{\Workers} \cup A_{\Firms} \\
\end{array}$
\end{center}

A {\it matching} appears in $D$ as an {\it independent set} of vertices, i.e., no two vertices contain the same worker or the same firm.
A {\it stable matching} appears in $D$ as a matching such that for every vertex $(\worker{},\firm{}) \in V$, either $(\worker{},\firm{})$ belongs to the matching or $(\worker{},\firm{})$ has an out-neighbour that does.

A matching is {\it maximal} if it is not a subset of any other matching.
A {\it maximum} matching is at least as large as every other matching.
It is almost immediate that a stable matching is maximal (since otherwise there must be an acceptable pair $(w, f)$ with both $w$ and $f$ unmatched).
At its core, our motivating question --- recall the opening paragraph --- concerns the potential difference in size between a (maximal) stable matching and a maximum one.


\section{The characterisation}\label{sec:characterisation}

Let $n$ be an integer greater than or equal to 2, and let us restrict attention to two-sided markets with $n$ workers, $n$ firms, and a stable matching of size $n$.
We note that this requirement is without loss of generality because the {\em normal form} of any matching market is always balanced (i.e., contains the same number of workers and firms) and contains only those participants that are matched in every stable matching.\footnote{Knowledge of the normal form strengthens the rural hospitals theorem \citep{McVitieWilson:1970:BIT,Roth:1984:JPE,Roth:1986:E}:~not only does it identify which participants appear in every stable matching, it also singles out many pairings that cannot be part of any stable matching amongst those that are always matched.}

The normal form is arrived at by running the {\em iterated deletion of unattractive alternatives} (IDUA) procedure \citep{BalinskiRatier:1997:,GutinNeary:2023:GEB}, a reduction that strips a matching market of all participants that do not appear in any stable matching.
The procedure works off the observation that some pairs in a matching market are never a barrier to stability.

\begin{definition}\label{def:reduce}
Given a matching digraph $D(P)$, we define $\reduce\big(D(P)\big)$ as the subdigraph of $D(P)$ with the following vertices (and any arcs incident on them) deleted.
\begin{enumerate}[label=(\roman*)]
\item\label{item:reduce1}
If no arc in $A_{\Firms}$ leaves $(\worker{}, \firm{}) \in V(D(P))$,  then delete all vertices
$(\worker{}, \firm{i})$ such that $(\worker{}, \firm{i}) (\worker{}, \firm{})  \in A_{\Workers}$.
\item\label{item:reduce2}
If no arc in  $A_{\Workers}$ leaves $(\worker{}, \firm{}) \in V(D(P))$,  then delete all vertices
$(\worker{i}, \firm{})$ such that $(\worker{i}, \firm{}) (\worker{}, \firm{})  \in A_{\Firms}$.
\end{enumerate}
\end{definition}

Consider part \ref{item:reduce1} of Definition~\ref{def:reduce}.
This simply notes that if worker $\worker{}$ is the favourite worker of firm $\firm{}$ then there cannot be a stable matching in which $\worker{}$ is paired with a firm that $\worker{}$ prefers less than $\firm{}$.
The reason is immediate:\ $(\worker{}, \firm{})$ would be a blocking pair against any such matching.
As a consequence of this, all pairs involving worker $\worker{}$ and a firm that it ranks below $\firm{}$ can be removed from consideration.
In \cite{GutinNeary:2023:GEB}, we referred to such pairs as {\em unattractive alternatives}. 
Part \ref{item:reduce2} is an analogous condition for firms.

The operation described in Definition~\ref{def:reduce} reduces the size of the matching digraph leaving behind a new smaller one.\footnote{Just as deleting a strictly dominated strategy from a strategic game removes all strategy profiles containing that strategy, deleting a vertex (pair) from a matching market removes all matchings that contain that pair.}
Furthermore, in fact almost by definition, the set of stable matchings has not changed.

Another key point is that there may exist vertices that are not removed by applying $\reduce$ to $D(P)$, but which are removed when applying $\reduce$ to $\reduce\big(D(P)\big)$.
That is:\ iterate.\footnote{This has parallels with how deleting dominated strategies in strategic games can render some originally undominated strategies dominated. In Section 2.3 of \cite{GutinNeary:2023:GEB} we detail some possible connections between the two concepts.}
Formally, the iterated application of $\reduce$ to $D(P)$ yields the following procedure.

\begin{definition}[The iterated deletion of unattractive alternatives (IDUA)]\label{def:IDUA}
Given a matching digraph $D(P) = \big(V, A_{\Workers} \cup A_{\Firms}\big)$, define $D^{0} = D(P)$, and for each $k \geq 1$, define $D^{k} = \reduce(D^{k-1})$. Finally, define the {\em normal form of} $P$, $D^{*}(P)$, as $D^{k}$ with minimum $k$ such that $D^{k} = \reduce(D^{k}),$ i.e., no further reductions will take place for such $D^{k}.$
\end{definition}

Lemma~2 in \cite{GutinNeary:2023:GEB} shows that $D^{*}(P)$ is uniquely defined which means that the order in which vertices satisfying the condition are chosen and subsequently deleted does not impact the final outcome.
That is, the normal form is uniquely defined.\footnote{The iterated deletion of strictly dominated strategies has the same property (in finite games), though the iterated deletion of weakly dominated strategies does not. See, for example, the game on page 1015 of \cite{KohlbergMertens:1986:E}.} 
Moreover, since applying $\reduce$ once does not change the set of stable matchings, repeated application of $\reduce$ doesn't either.
Given this, the set of stable matchings for the original market, $D(P)$, is unchanged and so can be computed from its normal form $D^*(P)$.

Below we define a class of matching markets, denoted by ${\cal G}_n$, that we show is precisely the class of markets that contain a stable matching of size $n$.
The procedure begins with a balanced matching market with a stable matching that is maximum and, from there, extends the preferences in a a manner such that the IDUA procedure of Definition~\ref{def:IDUA} would prune what is added.
That is, keeping the size of a stable matching held fixed, the extension allows for a larger individually rational matching.
Although, as per Theorem~\ref{thm:main}, this holds only up to a point.

\begin{definition}[The class of two-sided markets, ${\cal G}_n$]\label{def:Gn}
Let $P$ be any matching market with a set of $n$ workers, $W$, and a set of $n$ firms, $F$, that contains a stable matching, $\match$, of size $n$.
Add any number of new workers $W'$ and any number of new firms $F'$.
Now, for any $(w,f) \in \match$ and $(w',f)$ where $w'$ is in $W'$, we add an arc from $(w',f)$ to $(w,f)$ (that is, $f$ prefers $w$ over $w'$).
And for any $(w,f) \in \match$ and $(w,f')$ where $f'$ is in $F'$, we add an arc from $(w,f')$ to $(w,f)$ (that is, $w$ prefers $f$ over $f'$).
Finally, there are no vertices $(w',f')$ where $w' \in W'$ and $f' \in F'$, and all preferences that are not specified by $P$ or the rules above can be made arbitrarily.

Now, consider the class of matching markets generated by the above procedure for all markets with $n$ workers, $n$ firms, that contain a stable matching of size $n$.
This defines the class ${\cal G}_n$. 

\end{definition}

We now illustrate how the construction laid out in Definition~\ref{def:Gn} operates using an example.
To keep things as simple as possible, we choose $n=2$, and begin with a market, call it $P$, that has two workers, $w_1$ and $w_2$, and two firms, $f_1$ and $f_2$.
We assume that all parties view those on the other side as acceptable, and we further assume that both workers and both firms have opposite preference in that $w_1$ prefers firm $f_2$ over $f_1$ whereas worker $w_2$ ranks them differently, with a similar statement holding for the firms.
Preferences of this nature possess a {\em cycle} \citep{Chung:2000:GEB}, that appear in a matching digraph as a directed cycle.
The left hand side of Figure~\ref{fig:Nhalf2} presents the digraph of $P$ where workers index rows and firms index columns so that the vertex $(i,j)$ in the figure refers to the pair $(w_i, f_j)$.
There are two stable matchings both of size two, $\{(\worker{1},\firm{1}),(\worker{2},\firm{2})\}$ and $\{(\worker{1},\firm{2}),(\worker{2},\firm{1})\}$, with the worker-optimal one given by the {yellow} vertices and the firm-optimal one by the {blue} vertices.

Using $P$ as the launch point, we now construct a larger, four-by-four market, $R$.
Specifically, to $P$ we add two additional workers, $w_3$ and $w_4$, and two additional firms, $f_3$ and $f_4$, with preferences enriched in a way that adheres to the rules described in Definition~\ref{def:Gn}.
In the new, larger market, $R$, with digraph given in Figure~\ref{fig:Nhalf2}, the size of a stable matching has not changed as it remains at 2, whereas a new individually rational matching of size 4 now exists (given by the collection of red vertices).

\begin{figure}[h!]
\begin{center}

\tikzstyle{vertexDOT}=[scale=0.23,circle,draw,fill]
\tikzstyle{vertexY}=[circle,draw, top color=gray!10, bottom color=gray!40, minimum size=11pt, scale=0.5, inner sep=0.99pt]
\tikzstyle{vertexW}=[circle,draw, top color=gray!1, bottom color=gray!1, minimum size=11pt, scale=0.5, inner sep=0.99pt]
\tikzstyle{vertexQ}=[circle,dotted, draw, top color=gray!1, bottom color=gray!1, minimum size=11pt, scale=0.55, inner sep=0.99pt]

\WithColor{
\tikzstyle{vertexZ}=[circle,draw, top color=\matchingC{}!20, bottom color=\matchingC{}!50, minimum size=11pt, scale=0.5, inner sep=0.99pt]
\tikzstyle{vertexWF}=[circle,draw, top color=blue!70, bottom color=yellow!100, minimum size=11pt, scale=0.5, inner sep=0.99pt]
\tikzstyle{vertexF}=[circle,draw, top color=blue!50, bottom color=blue!50, minimum size=11pt, scale=0.5, inner sep=0.99pt]
\tikzstyle{vertexR}=[circle,draw, top color=red!20, bottom color=red!60, minimum size=11pt, scale=0.5, inner sep=0.99pt]
}

\WithoutColor{
\tikzstyle{vertexZ}=[rectangle,draw, top color=gray!1, bottom color=gray!1, minimum size=21pt, scale=0.6, inner sep=2.99pt]
\tikzstyle{vertexWF}=[circle,draw, top color=blue!70, bottom color=yellow!100, minimum size=11pt, scale=0.5, inner sep=0.99pt]
\tikzstyle{vertexF}=[rectangle,draw, top color=gray!80, bottom color=gray!50, minimum size=21pt, scale=0.6, inner sep=2.99pt]
\tikzstyle{vertexR}=[rectangle,draw, top color=black!70, bottom color=gray!70, minimum size=21pt, scale=0.6, inner sep=2.99pt]
}

\begin{tikzpicture}[scale=0.55]
\node [scale=1.2] at (2,6) {$P$};
\node [scale=1.2] at (2,4) {  };

\node [scale=0.9] at (1,12.5) {$f_1$};
\node [scale=0.9] at (3,12.5) {$f_2$};

\node [scale=0.9] at (-1,11) {$w_1$};
\node (x11) at (1,11) [vertexF] {$(1,1)$};
\node (x12) at (3,11) [vertexZ] {$(1,2)$};

\node [scale=0.9] at (-1,9) {$w_2$};
\node (x21) at (1,9) [vertexZ] {$(2,1)$};
\node (x22) at (3,9) [vertexF] {$(2,2)$};


\draw [->,line width=0.03cm] (x11) -- (x12);


\draw [->,line width=0.03cm] (x22) -- (x21);


\draw [->,line width=0.03cm] (x21) -- (x11);
\draw [->,line width=0.03cm] (x12) -- (x22);
\end{tikzpicture} \hspace{3cm}
\begin{tikzpicture}[scale=0.55]
\node [scale=1.2] at (7,4) {$R$};

\node [scale=0.9] at (-1,11) {$w_1$};
\node [scale=0.9] at (-1,9) {$w_2$};
\node [scale=0.9] at (-1,6) {$w_3$};
\node [scale=0.9] at (-1,4) {$w_4$};

\node [scale=0.9] at (1,12.5) {$f_1$};
\node [scale=0.9] at (3,12.5) {$f_2$};
\node [scale=0.9] at (6,12.5) {$f_3$};
\node [scale=0.9] at (8,12.5) {$f_4$};

\node (x11) at (1,11) [vertexY] {$(1,1)$};
\node (x12) at (3,11) [vertexZ] {$(1,2)$};
\node (x14) at (8,11) [vertexR] {$(1,4)$};

\node (x21) at (1,9) [vertexZ] {$(2,1)$};
\node (x22) at (3,9) [vertexY] {$(2,2)$};
\node (x23) at (6,9) [vertexR] {$(2,3)$};
\node (x24) at (8,9) [vertexY] {$(2,4)$};

\node (x32) at (3,6) [vertexR] {$(3,2)$};
\node (x41) at (1,4) [vertexR] {$(4,1)$};


\draw [->,line width=0.03cm] (x11) to [out=30,in=150]  (x14);
\draw [->,line width=0.03cm] (x14) -- (x12);


\draw [->,line width=0.03cm] (x22) -- (x23);
\draw [->,line width=0.03cm] (x23) -- (x24);
\draw [->,line width=0.03cm] (x24) to [out=210,in=330]  (x21);

\draw [->,line width=0.03cm] (x21) -- (x11);
\draw [->,line width=0.03cm] (x12) -- (x22);
\draw [->,line width=0.03cm] (x32) to [out=60,in=300]  (x12);
\draw [->,line width=0.03cm] (x41) -- (x21);

\draw [->,line width=0.03cm] (x14) -- (x24);

\end{tikzpicture}
\end{center}

\caption{A two-sided matching market, $P$,  with a stable matching of size two (depicted \WithColor{by the yellow vertices}\WithoutColor{by the white rectangles}) and a matching market, $R$, with a stable matching of size two and a matching of size 4 (depicted \WithColor{by the red vertices}\WithoutColor{by the dark rectangles}). Note that 
$\{(1,1),(2,2)\}$ is a stable matching in $P$ but is not in $R$.
All horizontal and vertical arcs that are implied by transitivity of preferences have been omitted for readability.}
\label{fig:Nhalf2}
\end{figure}

While the extensions permitted by Definition~\ref{def:Gn} will not alter the size of a stable matching, it is worth highlighting that certain extensions can impact the set of stable matchings.
For example, in Figure~\ref{fig:Nhalf2} the matching $\{(1,1),(2,2)\}$ was a stable matching in $P$ but is no longer stable when $P$ is extended to $R$.

We now show that the class ${\cal G}_n$ defined in Definition~\ref{def:Gn} precisely characterises all matching markets that contain a stable matching of size $n$.

\begin{lemma}\label{thm:characterize2}
A stable matching problem, $R$, contains a stable matching of size $n$ if and only if $R \in {\cal G}_n$.
\end{lemma}

\begin{proof}
Let $R$ be any two-sided matching market that contains a stable matching, $\match$, of size $n$.
Let $W_R$ denote the workers in $R$ and let $F_R$ denote the firms in $R$.
Let $W_\match$ be the workers and $F_\match$ the firms that are matched in $\match$.
Let $W'=W_R \setminus W_\match$ and let $F'=F_R \setminus W_\match$.
{Finally, denote by $D_R = \big(V(D_R), A(D_R)\big)$ the digraph constructed from market $R$.}

First we show that for every $(w,f) \in \match$ and every $f' \in F'$, either we have that $(w,f') \not\in V(D_R)$ or we have that $(w,f')(w,f) \in A(D_R)$.
For the sake of contradiction, assume that $(w,f') \in V(D_R)$ and  $(w,f)(w,f') \in A(D_R)$.
However, this implies that $(w,f')$ is a blocking pair for $\match$ (as $w$ prefers $f'$ over $f$ and $f'$ is not matched with any worker).
So, for every $(w,f) \in \match$ and every $f' \in F'$, we either have $(w,f') \not\in V(D_R)$ or $(w,f')(w,f) \in A(D_R)$.
An analagous argument shows that for every $(w,f) \in \match$ and every $w' \in W'$, we have either $(w',f) \not\in V(D_R)$ or $(w',f)(w,f) \in A(D_R)$.

Clearly $(w',f') \not\in V(D_R)$ for all $w' \in W'$ and $f' \in F'$, because such a vertex would be a blocking pair for $\match$, a contradiction to $\match$ being a stable matching in $R$. Therefore $R \in {\cal G}_n$.

We now show that every matching market in ${\cal G}_n$ has a stable matching of size $n$. 
Let $R \in {\cal G}_n$ be arbitrary.
Let $\match$ be the stable matching of $P$, which had size $n$, when we constructed $R$.
Since no pair in $P$ is a blocking pair for $\match$ and no pair including a worker or firm not in $P$ is a blocking pair for $\match$, it must be that $\match$ is a stable matching for $R$.
Therefore, $R$ has a stable matching of size $n$ as desired.
\end{proof}

If $n$ is odd, then $n/2$ is not an integer and subtlety arises because this cannot be the size of a stable matching.
However, Lemma~\ref{thm:characterize2} can be adapted as follows.
Let $n$ be an integer, odd or even, and define a class of markets, ${\cal F}_n$, as all markets in ${\cal G}_{\lceil n/2 \rceil}$ with a matching of size $n$.
In other words ${\cal F}_n$, is defined as follows.

\begin{definition}[The class of two-sided markets, ${\cal F}_n$]\label{def:Fn}
Let $P$ be any two sided market with $\lceil n/2 \rceil$ workers and $\lceil n/2 \rceil$ firms containing a stable matching, $\match$, of size $\lceil n/2 \rceil$.
Add at least $\lfloor n/2 \rfloor$ new workers $W'$ and at least $\lfloor n/2 \rfloor$ new firms $F'$ such that there is a matching of size $n$ in
the resulting market. 
Furthermore, for any $(w,f) \in S$ and $(w',f)$ where $w'$ is in $W'$, we add an arc from $(w',f)$ to $(w,f)$ (that is, $f$ prefers $w$ over $w'$).
Analogously, for any $(w,f) \in S$ and $(w,f')$ where $f'$ is in $F'$, we add an arc from $(w,f')$ to $(w,f)$ (that is, $w$ prefers $f$ over $f'$).
There are no vertices $(w',f')$ where $w' \in W'$ and $f' \in F'$ and all preferences that are not specified by $I$ or the rules above can be made arbitrarily.

Now, consider the class of matching markets generated by the above procedure for all markets with $n$ workers, $n$ firms, that contain a stable matching of size $n$.
This defines the class ${\cal F}_n$. 
\end{definition}

The following shows that the class of preferences ${\cal F}_n$ is precisely that which define matching markets where stability and efficiency --- that is, efficiency as measured by the size of largest individually rational matching --- are as far apart as possible.

\begin{theorem}\label{thm:characterize3}
A matching market, $R$, contains a matching of size $n$ and a stable matching of size $\lceil n/2 \rceil$ if and only if $R \in {\cal F}_n$.
\end{theorem}

\begin{proof}
By Lemma~\ref{thm:characterize2}, we note that every market in ${\cal F}_n$ contains a stable matching of size $\lceil n/2 \rceil$.
And by construction, all markets in ${\cal F}_n$ also contain a matching of size $n$. 

Now let $R$ be any market with a stable matching of size $\lceil n/2 \rceil$ and a matching of size $n$. 
By Lemma~\ref{thm:characterize2} we note that $R \in {\cal G}_{\lceil n/2 \rceil}$. This implies that
$R \in {\cal F}_n$ as $R$ contains a matching of size $n$. 
\end{proof}

We conclude this section by describing a natural class of preferences that satisfy the criteria.
This class could be described as having ``agreement at the top''.
By this we mean the following.
Let $W^*$ and $F^*$ be the set of workers and the set of firms that appear in a market's {\em normal form} and hence in every stable matching.
And suppose that all pairs in $W^* \times F^*$ are acceptable.
Now, if a worker $w$ in $W^*$ deems a newly added firm $f$ as acceptable, then $w$ also ranks $f$ below all firms in $F^*$.
A similar but opposite statement holds for newly added workers and how firms in $F^*$ rank them.

To give an example, let us extend the market $P$ in Figure~\ref{fig:Nhalf2} to market $R'$.
(This extension will be different as what was done to generate the market $R$ in Figure~\ref{fig:Nhalf2}.)
To generate the market $R'$ in Figure~\ref{fig:amend}, for all new workers and new firms, we insist that they are ranked by those who appear in every stable matching below these participants least favourite pairing over all stable matchings
This can be seen by observing that from every new vertex (pair) in the directed graph of $R$, there is an arc into the least preferred participant in $P$.

Note further that this additional requirement will not alter the set of stable matchings.
That is, the normal form of the market $P$ and the market $R'$ will be same.

Finally, note that more workers and more firms could be added without changing the set of stable matchings nor the size of a maximum matching.
This will happen as long as all new workers are ranked by firms $f_1$ and $f_2$ below $w_1$ and $w_2$, with a similar condition for the newly added firms.

\begin{figure}[h!]
\begin{center}

\tikzstyle{vertexDOT}=[scale=0.23,circle,draw,fill]
\tikzstyle{vertexY}=[circle,draw, top color=gray!10, bottom color=gray!40, minimum size=11pt, scale=0.5, inner sep=0.99pt]
\tikzstyle{vertexW}=[circle,draw, top color=gray!1, bottom color=gray!1, minimum size=11pt, scale=0.5, inner sep=0.99pt]
\tikzstyle{vertexQ}=[circle,dotted, draw, top color=gray!1, bottom color=gray!1, minimum size=11pt, scale=0.55, inner sep=0.99pt]

\WithColor{
\tikzstyle{vertexZ}=[circle,draw, top color=\matchingC{}!20, bottom color=\matchingC{}!50, minimum size=11pt, scale=0.5, inner sep=0.99pt]
\tikzstyle{vertexWF}=[circle,draw, top color=blue!70, bottom color=yellow!100, minimum size=11pt, scale=0.5, inner sep=0.99pt]
\tikzstyle{vertexF}=[circle,draw, top color=blue!50, bottom color=blue!50, minimum size=11pt, scale=0.5, inner sep=0.99pt]
\tikzstyle{vertexR}=[circle,draw, top color=red!20, bottom color=red!60, minimum size=11pt, scale=0.5, inner sep=0.99pt]
}

\WithoutColor{
\tikzstyle{vertexZ}=[rectangle,draw, top color=gray!1, bottom color=gray!1, minimum size=21pt, scale=0.6, inner sep=2.99pt]
\tikzstyle{vertexWF}=[circle,draw, top color=blue!70, bottom color=yellow!100, minimum size=11pt, scale=0.5, inner sep=0.99pt]
\tikzstyle{vertexF}=[rectangle,draw, top color=gray!80, bottom color=gray!50, minimum size=21pt, scale=0.6, inner sep=2.99pt]
\tikzstyle{vertexR}=[rectangle,draw, top color=black!70, bottom color=gray!70, minimum size=21pt, scale=0.6, inner sep=2.99pt]
}

\begin{tikzpicture}[scale=0.55]
\node [scale=1.2] at (2,6) {$P$};
\node [scale=1.2] at (2,4) {  };

\node [scale=0.9] at (-1,11) {$w_1$};
\node (x11) at (1,11) [vertexF] {$(1,1)$};
\node (x12) at (3,11) [vertexZ] {$(1,2)$};

\node [scale=0.9] at (-1,9) {$w_2$};
\node (x21) at (1,9) [vertexZ] {$(2,1)$};
\node (x22) at (3,9) [vertexF] {$(2,2)$};

\node [scale=0.9] at (1,12.5) {$f_1$};
\node [scale=0.9] at (3,12.5) {$f_2$};


\draw [->,line width=0.03cm] (x11) -- (x12);


\draw [->,line width=0.03cm] (x22) -- (x21);


\draw [->,line width=0.03cm] (x21) -- (x11);
\draw [->,line width=0.03cm] (x12) -- (x22);
\end{tikzpicture} \hspace{3cm}
\begin{tikzpicture}[scale=0.55]
\node [scale=1.2] at (7,4) {$R'$};

\node [scale=0.9] at (1,12.5) {$f_1$};
\node [scale=0.9] at (3,12.5) {$f_2$};
\node [scale=0.9] at (6,12.5) {$f_3$};
\node [scale=0.9] at (8,12.5) {$f_4$};

\node [scale=0.9] at (-1,11) {$w_1$};
\node (x11) at (1,11) [vertexF] {$(1,1)$};
\node (x12) at (3,11) [vertexZ] {$(1,2)$};
\node (x13) at (6,11) [vertexY] {$(1,3)$};
\node (x14) at (8,11) [vertexR] {$(1,4)$};

\node [scale=0.9] at (-1,9) {$w_2$};
\node (x21) at (1,9) [vertexZ] {$(2,1)$};
\node (x22) at (3,9) [vertexF] {$(2,2)$};
\node (x23) at (6,9) [vertexR] {$(2,3)$};
\node (x24) at (8,9) [vertexY] {$(2,4)$};

\node [scale=0.9] at (-1,6) {$w_3$};
\node (x31) at (1,6) [vertexY] {$(3,2)$};
\node (x32) at (3,6) [vertexR] {$(3,2)$};

\node [scale=0.9] at (-1,4) {$w_4$};
\node (x41) at (1,4) [vertexR] {$(4,1)$};
\node (x42) at (3,4) [vertexY] {$(4,2)$};

\draw [->,line width=0.03cm] (x11) -- (x12);
\draw [->,line width=0.03cm] (x14) to [out=150,in=30]  (x11);


\draw [->,line width=0.03cm] (x23) -- (x22);
\draw [->,line width=0.03cm] (x24) -- (x23);

\draw [->,line width=0.03cm] (x21) -- (x11);
\draw [->,line width=0.03cm] (x12) -- (x22);
\draw [->,line width=0.03cm] (x22) -- (x21);
\draw [->,line width=0.03cm] (x32) to [out=60,in=300]  (x12);
\draw [->,line width=0.03cm] (x42) to (x32);
\draw [->,line width=0.03cm] (x31) -- (x21);
\draw [->,line width=0.03cm] (x41) -- (x31);

\draw [->,line width=0.03cm] (x13) -- (x14);
\draw [->,line width=0.03cm] (x14) -- (x24);

\end{tikzpicture}
\end{center}

\caption{On the left is the two-sided matching market, $P$, from Figure~\ref{fig:Nhalf2}. On the right is a matching market, $R'$, with a stable matching of size two and a matching of size 4 (depicted \WithColor{by the red vertices}\WithoutColor{by the dark rectangles}). Note that both stable matchings in $P$ are also stable matchings in $R'$. As in Figure~\ref{fig:Nhalf2}, arcs implied by transitivity have been omitted to avoid clutter.}
\label{fig:amend}
\end{figure}


\section{Stable pairings and maximum employment\protect\footnote{We thank two anonymous referees for suggesting the analysis of this section.}}\label{sec:relationship}

In the market with extended preferences in the right hand panel Figure~\ref{fig:Nhalf2}, none of the pairs contained in a stable matching also appeared in the maximum matching (the same holds for the extended market in Figure~\ref{fig:amend}).
Is this an artefact of this particular market, or can something general be said?

It turns out that there is a simple, albeit we believe somewhat surprising, relationship between the overlap in pairs that appear in (maximal) stable matchings and maximum matchings.
Theorem~\ref{thm:overlap} makes this precise.
We first state and prove the result, and we then discuss its consequences.

\begin{theorem}\label{thm:overlap}
Given a matching market $P$, let $\mu^*$ denote a stable matching, let
$\mu$ denote any matching (not necessarily stable), and let $r$ denote the number of pairs that $\mu^*$ and $\mu$ have in common.
Then $|\mu^*| \geq \frac{|\mu|+r}{2}$.
\end{theorem}

\begin{proof}
Let $M=\mu^* \cap \mu$.
That is, $M$ contains the $r$ pairs that $\mu^*$ and $\mu$ have in common.
Remove all workers and firms in $M$ and call the remaining market $P'$.
Then the remaining pairs in $\mu^* \setminus M$ form a stable matching in $P'$ and
$\mu \setminus M$ forms a matching in $P'$.
By Theorem~\ref{thm:main} we note that
$|\mu^* \setminus M| \geq \ceiling{\frac{|\mu \setminus M|}{2}}$.
This implies that the following holds.

\[
|\mu^*|-r \geq \left\lceil \frac{|\mu|-r}{2} \right\rceil  \geq \frac{|\mu|-r}{2}
\]

This in turn implies that $|\mu^*| \geq \frac{|\mu|+r}{2}$ as desired.
\end{proof}

By rearranging the inequality in Theorem~\ref{thm:overlap}, we get that $r \leq 2|\mu^*| - |\mu|$, thereby providing an upper bound on the number of equilibrium pairings that can appear in $\mu$.
While the result holds for any matching $\mu$, when $|\mu|$ is maximum (i.e., $\mu$ is a largest individually rational matching), we can identify the minimum number of equilibrium pairings that must be ``sacrificed'' if maximising the employment rate is the objective.
In particular, when the difference in size between a stable matching and an individually rational matching is maximised, and as per Theorem~\ref{thm:main} this occurs when $|\mu| = 2|\mu^*|$, it must be that no equilibrium pairings appear in a maximum matching.

The above confirms that there can be an additional {\it cost of full employment} that occurs at a more granular level:~not only can a demand of filling all vacancies come at the cost of stability, it can also require forfeiting some equilibrium pairings.



\clearpage

\bibliographystyle{plainnat}
\bibliography{matchSize.bib}


\end{document}